%% file: main.tex
\documentclass[11pt]{article}

\pdfoutput=1

\usepackage{fullpage}
\usepackage[utf8]{inputenc}
\usepackage{algorithm,algpseudocode,amsmath, amssymb,amsthm,authblk,braket,latexsym,mathtools,optidef}
\usepackage[pdftex]{graphicx}
\usepackage{hyperref}
\usepackage[dvipsnames]{xcolor}
\usepackage{tcolorbox}

\newcommand{\NP}{\mathsf{NP}}
\newcommand{\PLang}{\mathsf{P}}

\newcommand{\AM}{\mathsf{AM}}

\newcommand{\MIP}{\mathsf{MIP}}

\newcommand{\RE}{\mathsf{RE}}

\newcommand{\NEXP}{\mathsf{NEXP}}
\newcommand{\PCP}{\mathsf{PCP}}
\newcommand{\IP}{\mathsf{IP}}
\newcommand{\PSPACE}{\mathsf{PSPACE}}

\newcommand{\poly}{\mathrm{poly}}
\newcommand{\polylog}{\mathrm{polylog}}

\newcommand{\cX}{\mathcal{X}}
\newcommand{\cY}{\mathcal{Y}}
\newcommand{\cA}{\mathcal{A}}
\newcommand{\cB}{\mathcal{B}}
\newcommand{\sV}{\mathsf{V}}

\newcommand{\val}{\mathsf{val}}
\newcommand{\eps}{\varepsilon}

\usepackage[capitalise,nameinlink]{cleveref}
\hypersetup{colorlinks={true},linkcolor={blue},citecolor=red}
\usepackage{thmtools}

\newtheorem{theorem}{Theorem}[section]
\newtheorem{definition}[theorem]{Definition}
\newtheorem{lemma}[theorem]{Lemma}
\newtheorem{claim}[theorem]{Claim}
\newtheorem{fact}[theorem]{Fact}
\newtheorem{corollary}[theorem]{Corollary}
\newtheorem{conjecture}[theorem]{Conjecture}

\title{Multi-Prover Interactive Proof Systems with Leakage}

\author[1]{Vahid R. Asadi\thanks{vrasadi@nii.ac.jp}}
\author[2]{Atsuya Hasegawa\thanks{atsuya.hasegawa@math.nagoya-u.ac.jp}}
\author[2]{Fran{\c{c}}ois Le Gall\thanks{legall@math.nagoya-u.ac.jp}}

\affil[1]{\textit{National Institute of Informatics, Japan}}
\affil[2]{\textit{Graduate School of Mathematics, Nagoya University, Japan}}

\date{}

\begin{document}

\maketitle

\begin{abstract}
    It is known that there exist multi-prover interactive protocols ($\mathsf{MIP}$ protocols) for the complexity class $\mathsf{NEXP}$, succinct $\mathsf{MIP}$ protocols for $\mathsf{NP}$ and multi-prover interactive protocols with shared entanglement ($\mathsf{MIP}^\ast$ protocols) for $\mathsf{RE}$. This extraordinary power of multi-prover interactive proof systems comes from the assumption that provers do not communicate with each other during the protocols. If they are allowed to communicate freely, the setting is the same as in the single-prover case, and the computational power of the system becomes significantly weaker.

    In this paper, we investigate for the first time the setting where communication (i.e., leakage of information) between provers is allowed but bounded. We introduce two techniques to approach this question and show that multi-prover interactive proof systems are robust against some amount of leakage. Our first technique is based on parallel repetition theorems. We apply it to show that for any polynomial $p$, we can construct two-prover one-round $\mathsf{MIP}$ and $\mathsf{MIP}^\ast$ protocols for $\mathsf{NEXP}$ and $\mathsf{RE}$, respectively, that are robust against $p(n)$ bits of leakage. We further derive our second technique to convert any low-soundness PCP construction to a two-prover one-round $\mathsf{MIP}$ protocol for $\mathsf{NP}$ robust against leakage. We also discuss the relation between robustness against leakage in multi-prover interactive proof systems and the Sliding Scale Conjecture in the PCP literature. 
\end{abstract}

\clearpage

\tableofcontents

\clearpage

\input{intro}
\input{preliminary}
\input{proof}
\input{summary}

\section*{Acknowledgments}
VRA thanks Shuichi Hirahara and Nobutaka Shimizu for helpful discussions. This research was done when VRA was a visiting researcher at Nagoya University, and a graduate student at the University of Waterloo. AH thanks Srijita Kundu for useful discussions. The authors also thank the anonymous referees for their helpful comments and suggestions.

VRA is supported by a JSPS Postdoctoral Fellowship for Research in Japan. AH is supported by JSPS KAKENHI grant No.~24H00071, 25K24674, 25K24465. FLG is supported by JSPS KAKENHI grant No.~24H00071, 25K24674, 25K24465, MEXT Q-LEAP grant No.~JPMXS0120319794, JST ASPIRE grant No.~JPMJAP2302 and JST CREST grant No.~JPMJCR24I4.

\bibliographystyle{alpha}
\bibliography{ref}

\end{document}

%% file: intro.tex
\section{Introduction}

\subsection{Background}
Multi-prover interactive proof systems ($\MIP$s) were first introduced by Ben-Or et al.\! \cite{ben-or1988multi} to remove intractability assumptions from zero-knowledge proof systems \cite{goldwasser1989knowledge,goldreich1991proofs}. In a multi-prover interactive proof system, there are multiple provers who interact with a verifier. The provers cannot communicate with each other during the execution of the protocol, and thus the verifier can ``cross-check'' their assertions. In 1990, Babai, Fortnow and Lund \cite{babai1991mip} showed $\MIP=\NEXP$, where $\NEXP$ is a set of problems recognized by non-deterministic exponential-time Turing machines. In $\MIP$ protocols, the polynomial-time verifier checks exponentially large proofs, and the result of $\MIP=\NEXP$ led to the discovery of the celebrated PCP theorem \cite{arora1998probabilistic,arora1998proof}.

The class $\MIP^*$, which is a variant of $\MIP$ where provers can share entanglement,  was first introduced by Cleve et al.\! \cite{cleve2004consequences} who derived a correspondence with quantum non-local games. Since in non-local games, quantum values are larger than classical values (i.e., $\omega_q(G) \geq \omega_c(G)$) and thus soundness in $\MIP$ protocols does not imply one in $\MIP^*$ protocols, the relation between $\MIP$ and $\MIP^*$ was mysterious. After an extensive effort to investigate the power of $\MIP^*$ (e.g., \cite{kempe2011entangled,ito2012multi,reichardt2013classical,natarajan2019neexp}), Ji et al.\! \cite{ji2020mip*} proved $\MIP^*=\RE$. The $\RE$-complete problem is the Halting Problem which is undecidable by any Turing machine, and shared entanglement between the provers gives extremely strong power to the multi-prover interactive proof systems.

The extraordinary power of multi-prover interactive proof systems comes from the canonical assumption that provers do not communicate with each other during the execution of the protocol\footnote{Before the protocol starts, they can talk freely and share some strategy.}. If provers are allowed to communicate freely, the proof system is the same as in the case of the single prover, and the power of the proof system becomes much weaker because $\IP=\PSPACE$ \cite{shamir1992ip}. If we limit the number of rounds to a constant, it is still the case that there exists a two-prover one-round $\MIP$ (resp. $\MIP^*$) protocol for $\NEXP$ (resp. $\RE$). However, in the single-prover setting ($\IP$), the class is equal to $\AM$ \cite{goldwasser1986private,babai1988arthur}, which is much weaker than the multi-prover case.

An important line of research in the study of multi-prover interactive proofs is related to the problem of reducing the soundness error to any arbitrary small value. This question motivated the study of the parallel repetition theorem \cite{raz1998parallel, hol2009parallel} in which the goal is to reduce the soundness error to any arbitrary value by asking the provers to play multiple instances of the same game simultaneously. It was first shown by \cite{raz1998parallel} that the soundness error decays exponentially with respect to the number of copies of the game that are played in parallel or at the same time.

The parallel repetition theorem also plays an important role in the study of probabilistically checkable proofs (PCPs), where it is used to reduce the soundness error at the cost of increasing the proof length and alphabet size.
However, parallel repetition is not the only approach. Gap amplification \cite{dinur2007pcp,bafna2025quasi} provides a more efficient way to increase the unsatisfiability gap, thereby reducing the soundness error with smaller overhead.
Moreover, similar effects can be achieved through careful composition of PCPs \cite{moshkovitz2008two,dinur2015lowerrorPCP}.

In this work, we formally show a connection between the no-communication assumption in multi-prover interactive proofs, and the parallel repetition and the low soundness error for PCPs. More precisely, we will use the latter to relax the former and show that if the soundness error is chosen to be a carefully small value, one can allow some bits\footnote{
For $\MIP^*$ protocols, quantum communication (qubits) can always be replaced by classical communication (bits) via teleportation \cite{bennett1993teleporting} since the provers share entanglement. In this paper, we thus only consider classical communication between provers, even when discussing $\MIP^*$ protocols.} of communication between the provers.

\subsection{Our contribution}

As we have discussed, there are huge gaps in the computational power of multi-prover interactive proof systems with and without communication between provers. Moreover, when we want to realize these systems physically, say for some cryptographic applications, enforcing physical separation may be considered an unrealistic assumption, and in practice, 
a small amount of information could leak from the devices corresponding to the provers.

These facts motivate us to ask the following question:
\begin{quote}
    \emph{Assume some bounded amount of communication is allowed between the provers. Then, what is the computational power of the multi-prover interactive proof systems?}
\end{quote}

Note that it is necessary to assume there is a meaningful bound on the amount of leakage between the provers. In particular, if the amount of leakage is greater than the size of the questions the verifier asks, then trivially the power of the protocol comes down to the power of single-prover interactive proofs. 

In this work, we will introduce two techniques to answer this question.
\begin{itemize}
    \item Our first tool, which uses the power of parallel repetition theorems, can be used to show that it is indeed the case that we can allow the provers to have bounded communication, while preserving the power of multi-prover settings. Of course, there will be a cost to pay, and it shows up in the size of the question and answer sets of the non-local game, but as we will discuss shortly, the blow-up in the question and answer size is moderate.
    \item Our second technique exploits low-soundness PCP constructions (achieved from gap amplification), and we show how to convert them to leakage-resilient succinct $\MIP$ protocols for $\NP$. We achieve better parameters for question, answer, leakage size, and soundness, although this comes at the cost of additional assumptions or restrictions.
\end{itemize}
Our techniques allow us to show that the extraordinary power of multi-prover interactive proof systems still holds even if we weaken the assumption of no communication.

This, to the best of our knowledge, is the first time one utilizes the power of soundness error reductions towards relaxing the no-communication assumption in the multi-prover interactive proof systems. We will also show connections between the possibility of certain multi-prover interactive proof systems with leakage and the Sliding Scale Conjecture in the PCP literature, with the hope that one can find resolutions for either using techniques from the other.

\subsection{Technical overview}\label{sec:overview}
For the rest of this technical overview, we fix the following notation. We denote by $\MIP[q,a,\ell]$ (resp. $\MIP^*[q,a,\ell]$) the class of problems decided by a two-prover one-round $\MIP$ (resp. $\MIP^*$) protocol whose question length is $q$ and answer length is $a$, and the amount of allowed communication between the provers is $\ell$ (see \cref{sec:def of MIP with leakage} for formal definitions). Here, $n$ always denotes the instance size of the problem, and for simplicity, we assume the alphabet is $\{0,1\}$.

\subsubsection{Techniques based on parallel repetition}

Our first contribution is to show how to use parallel repetition theorems for classical and quantum non-local games to remove the no-communication assumption between provers in a multi-prover interactive proof system. 

We utilize this technique to first show the following result for $\MIP$.

\begin{theorem}[Informal version of \cref{thm:mip leakage}]\label{thm:main classical informal theorem}
Even if any $\poly(n)$ bits of communication is allowed between the provers, there exists a two-prover one-round $\MIP$ protocol for $\NEXP$.
\end{theorem}

Note that in the protocol above, the question size is some asymptotically larger polynomial than the leakage size, and the answer size is a sufficiently large linear multiple of the leakage size.

Second, considering the parallel repetition of quantum non-local games, we show the following result for $\MIP^*$.

\begin{theorem}[Informal version of \cref{thm:mip* with leakage}]\label{thm:main quantum informal theorem}
Even if any $\poly(n)$ bits of communication is allowed between the provers, there exists a two-prover one-round $\MIP^*$ protocol for $\RE$.
\end{theorem}

In the protocol we show above, the question size is a sufficiently large linear multiple of the leakage size, and the answer size is larger than the leakage size and asymptotically the same polynomial up to some $\polylog(n)$ factor.

Our results imply that any $\poly(n)$ amount of leakage between provers in a multi-prover interactive proof can be allowed without changing their computational power.

\paragraph{Proof strategies:} 
Given the correspondence between $\MIP$ and $\MIP^*$ protocols and non-local games, we can consider an $N$-times parallel repetition of a two-prover one-round $\MIP$ (resp. $\MIP^*$) protocol for $\NEXP$ (resp. $\RE$). Roughly speaking, the parallel repetition of a non-local game $G$ is a game where the parties try to win $N$ copies of $G$ simultaneously. The parallel repetition game is denoted by $G^{\otimes N}$. Suppose that the amount of communication allowed between the provers is sublinear in $N$. Then, intuitively, it is hard for the cheating provers to win all $N$ games. This can be formalized by a direct product theorem in communication complexity, which has a long history of works (e.g., \cite{razborov1992distributional,chakrabarti2001informational,jain2005prior,braverman2013direct}).

For classical communication complexity, it is known that a direct product theorem follows from the parallel repetition theorem for the classical value, first shown by Raz \cite{raz1998parallel}.
It is also known that there exists a $\MIP$ protocol for $\NP$ and $\NEXP$ whose answer size is constant from the PCP theorem \cite{arora1998probabilistic,arora1998proof,dinur2007pcp,moshkovitz2008two,bafna2025quasi}. Therefore, by considering a sufficiently large $\poly(n)$ times parallel repetition of $\MIP$ protocols for $\NEXP$ with perfect completeness, we have $\MIP[\poly(n),\poly(n),\poly(n)] \supseteq \NEXP$, which proves \cref{thm:main classical informal theorem}.

For the entanglement-assisted interactive quantum communication complexity, Jain and Kundu recently \cite{jain2022direct} showed a direct product theorem for the parallel repetition of non-local games whose question distribution is product (see \cref{lem:quantum dpt} for a formal description).

In a non-local game, if the distribution of questions is a product distribution, we denote such games by free games. Natarajan and Zhang \cite{natarajan2023quantum} showed that there exists a succinct quantum free game protocol for the Halting Problem, which meets the conditions we want. Considering $\poly(n)$-times parallel repetition of the protocol in \cite{natarajan2023quantum}, and combining it with the direct product theorem in \cite{jain2022direct}, we have $\MIP^*[\poly(n),\poly(n),\poly(n)] \supseteq \RE$, giving us \cref{thm:main quantum informal theorem}.

We will formalize these proofs in \cref{sec:MIP,sec:MIP*}.
 
\subsubsection{Techniques based on low-soundness PCPs}

As our next contribution, we develop a new technique where we show how to convert any PCP construction with small soundness error into a $\MIP$ protocol resilient against leakage. See \cref{thm:CSP-to-MIP} for a formal statement. For this technique, we need the additional assumption that the leakage between the provers is one-way, meaning that only the first prover can send messages to the second prover. Another way of viewing this scenario is to assume that the verifier sends a question to the first prover, receives the answer, and all the communication between the provers happens at this stage, before the second prover receives its question from the verifier. Therefore, any communication from the second malicious prover to the first is not helpful in winning the game.

We then use this tool to show a succinct $\MIP$ protocol for $\NP$ that is sound against constant bits of leakage. Here succinct means that the question and answer size are small, and thus while $\NP \subset \NEXP$, this result is not superseded by \cref{thm:main classical informal theorem}. An informal statement of our result is the following.

\begin{theorem}[Informal version of  \cref{thm:lbitPCP}]\label{thm:lbitPCP-informal}
    There exists a two-prover one-round $\MIP$ protocol for $\NP$ where the question size is $2\log n + O(\log \log n)$, the answer size is $O(1)$, and the system is sound against $O(1)$ bits of one-way leakage.
\end{theorem}

Note that in this setting, if we allow the provers to communicate freely (in other words, having a single prover while the questions are $O(\log n)$ in length), the power of the protocol reduces to $\PLang$, essentially rendering the verification protocol useless. What we show is that we can still maintain the power of verifying $\NP$ while the provers are allowed to communicate a constant number of bits, and the overhead to pay is only an additive $O(\log \log n)$ term in the length of the question. 

We further show that there exists a $\MIP$ protocol for $\NP$ that is sound against $O(\log n)$ bits of leakage, but our result comes at the cost of super-constant soundness error.

\begin{theorem}[Informal version of \cref{thm:logn-leakage}]\label{thm:logn-leakage-informal}
    For any positive constant $\ell$, there exists a two-prover one-round $\MIP$ protocol allowing $\ell\log n$ bits of one-way leakage for $\NP$ where the question size is $O(\log n)$, the answer size is $\log n\cdot \poly(\log\log n)$, and the system has soundness $1-1/\poly(\log\log n)$.
\end{theorem}

\paragraph{Proof strategies:} We start by showing a correspondence between low-soundness PCPs and the construction of $\MIP$s in which the first prover is allowed to communicate a bounded number of bits to the second prover. Roughly speaking, by utilizing a PCP construction in which the number of satisfied constraints in a no instance is very low, we can make sure that even if the provers decide to change their strategy mid-protocol, there are still many constraints that remain unsatisfied even with a new cheating strategy. For simplicity, suppose that the provers are only able to communicate 1 bit. This corresponds to the first prover specifying which of the two strategies the second prover has to respond with (one strategy corresponds to bit 0, and the other one to bit 1). Therefore, if we can amplify the soundness gap to a large enough constant (say, something less than $1/10$), then for any two different assignments they agree on before the protocol starts, there are still many constraints that remain unsatisfied by both. Hence, if the verifier is lucky enough to sample one of those, it can catch the cheating provers. We will formalize this intuition in \cref{sec:PCP2025}.

It is important to note that it is not clear to us whether Dinur's gap amplification theorem \cite{dinur2007pcp} can be used to obtain such a result from our techniques. The main reason is that for our proof to go through, we need an \emph{arbitrarily} small (but still constant) soundness error, which is something that Dinur's protocol cannot achieve, as observed by \cite{bogdanov2005gap}. The original PCP construction of \cite{arora1998proof, arora1998probabilistic} is also not useful due to the same reason. However, by using the constructions in \cite{moshkovitz2008two,bafna2025quasi}, we can achieve a very succinct $\MIP$ protocol for $\NP$ that is sound against leakage. This is because these constructions have the property of achieving an arbitrarily small soundness error by paying a moderate cost in the alphabet size and the proof length. To prove \cref{thm:lbitPCP-informal}, we make use of the recent breakthrough of \cite{bafna2025quasi}. We can have a similar result with a question length of $2\log n + O(\sqrt{\log n})$ from \cite{moshkovitz2008two}, which is larger than our result of $2\log n + O(\log \log n)$. To prove \cref{thm:logn-leakage-informal}, we will make use of a similar argument, but we need to utilize a PCP construction that achieves $1/\poly(n)$ soundness error. To this end, we will use the construction of Dinur et al.\! \cite{dinur2015lowerrorPCP} which has $\poly(\log\log n)$ query complexity, but the soundness error vanishes inverse polynomially in $n$. This will allow us to afford $O(\log n)$ bits of leakage, without incurring too much overhead in the question size. However, the super-constant query complexity causes the soundness error to not be a constant anymore.

This result also has interesting connections to the Sliding Scale Conjecture in the PCP literature. We will state it below and refer the interested reader to the survey of Moshkovitz \cite{moshkovitz2019slidingscale} for more information.

\begin{conjecture}[Sliding Scale Conjecture \cite{bellare1993efficientpcp}]\label{conj:sliding-scale}
    For any $1/\poly(n) \leq s <1$ and for all languages in $\NP$, there exists a PCP verifier $V$ with perfect completeness and soundness $s$ such that $V$ tosses $\log n$ random coins and makes a constant number of queries to a proof over an alphabet of size $1/\poly(s)$.
\end{conjecture}

Assuming that this conjecture is true, we can use it in our technique to improve the soundness error in \cref{thm:logn-leakage-informal} to a constant, and keep the answer size $O(\log n)$. Hence, it is possible to show that ruling out the existence of such succinct $\MIP$ protocols with $O(\log n)$ bits of leakage for $\NP$ can give a resolution for \cref{conj:sliding-scale}. It is an interesting open problem to figure out whether we can achieve constant soundness using other techniques. We will discuss more about open problems in \cref{sec:open-problems}.

\paragraph{Comparison of the two approaches:} We would like to emphasize that the two approaches that we present in this work are not entirely unrelated. Indeed, the key idea behind both is to reduce the soundness error in a multi-prover proof system to a carefully chosen small constant and utilize this to show that even a bounded amount of leakage does not help the cheating provers to succeed. Note that to obtain a succinct $\MIP$ protocol for $\NP$, we can allow leakage directly from a part of our first technique by utilizing a 2-query low-soundness PCPs \cite{moshkovitz2008two,bafna2025quasi}\footnote{We cannot obtain similar results using PCPs whose soundness error is at least 1/2 \cite{arora1998proof,arora1998probabilistic,dinur2007pcp} for the same reason we discussed earlier.} (see \cref{cor:immediate} for more details). In this case, the parameters are even better than what we have in \cref{thm:lbitPCP-informal} from our second technique. However, the strength of our second technique is that we can deal with general low-soundness PCPs whose number of queries is more than 2, and in particular, this allows us to achieve interesting parameters such as \cref{thm:logn-leakage-informal}.

Despite the similarity, there are differences in the power of these techniques. First, it is important to note that the second technique is only applicable when we make the additional assumption that the leakage is one-way. Also, it is not clear how to apply the second technique to achieve a result similar to the one presented in \cref{thm:main quantum informal theorem}. It could be the case that one achieves a similar result assuming the quantum PCP conjecture is true, but we will not discuss this possibility here and leave it to future work. Note that although the first technique is more powerful in capturing different scenarios, it inherently suffers from a polynomial blow-up in the game size. Our second technique suggests that this blow-up is not optimal.

\subsection{Related works}

Some decays of values by parallel repetition (the parallel repetition theorem) for classical values (e.g., \cite{raz1998parallel,hol2009parallel,rao2011parallel}) and quantum values (e.g., \cite{cleve2008perfect,kempe2011parallel,jain2014parallel,yuen2016parallel,bavarian2022anchored}) are known. 

We usually prohibit communication between the provers during the protocol. One exception, to our knowledge, is a model considered by Ben-Or, Hassidim and Pilpel \cite{ben-or2014quantum}. They considered a variant of quantum $\MIP$, where the provers do not share entanglement, the communication between the verifier and the provers is quantum, but the provers are allowed to have unlimited classical communication between them. They showed that any language in $\NEXP$ can be recognized in the model, which is the same as $\MIP$.

Communication complexity of parallel repetition of non-local games has been considered. Jain and Kundu \cite{jain2022direct} showed a direct product theorem for the entanglement-assisted communication complexity, and its application for device-independent quantum key distribution \cite{jain2020parallel}. Hasegawa, Le Gall and  Modanese \cite{hasegawa2025maximum} considered the parallel repetition of Magic Square games \cite{mermin1990simple,peres1990incompatible,aravind2002simple}, and showed that even if $o(n)$ EPR pairs are allowed to share prior to a communication protocol, $\Omega(n)$ qubits of communication is required to solve all Magic Square games with high probability.

%% file: preliminary.tex
\section{Preliminaries}\label{sec:prel}

We assume that the readers are familiar with basic concepts of computational complexity theory, quantum computing and information theory. We refer to \cite{arora2009computational} for a reference in computational complexity, \cite{buhrman2010nonlocality} in non-locality and communication complexity, and \cite{nielsen2010quantum,watrous2018theory,deWolf2019quantum} in quantum computing and information theory.

\subsection{Non-local games and parallel repetition theorem}
We say for a non-local game $G = (\pi, \cX \times \cY, \cA \times \cB, \mathsf{V})$, $\max\{\log |\cX|,\log|\cY|\}$ is the question size and $\max\{\log|\cA|,\log|\cB|\}$ is the answer size. The distribution $\pi$ is product if there are distributions $\pi_\cX$ , $\pi_\cY$ on $\cX$, $\cY$ respectively such that $\forall (x, y) \in \cX \times \cY : \pi(x, y) =\pi_\cX (x) \cdot \pi_\cY (y)$.

The classical value of a game $G = (\pi, \cX \times \cY, \cA \times \cB, \mathsf{V})$ is the maximum probability with which Alice and Bob can win $G$,
ranging over all purely classical strategies. The classical value of a game $G$ will be denoted $\omega_c(G)$.
A deterministic strategy is a restricted type of classical strategy in which $a$ and $b$ are simply functions of $x$ and $y$, respectively. The classical value of a game is always obtained by some deterministic strategy given that any probabilistic strategy can be expressed as a convex combination of deterministic strategies. In other words, it holds that
\[
    \omega_c(G) = \max_{a,b} \sum_{x,y} \pi(x,y) \mathsf{V}(a(x),b(y)|x,y) \enspace,
\]
where the maximum is over all functions $a : \cX \rightarrow \cA$ and $b : \cY \rightarrow \cB$.

A quantum strategy for a game $G = (\pi, \cX \times \cY, \cA \times \cB, \mathsf{V})$ consists of an initial bipartite state $\ket{\psi} \in \mathcal{H}_A \otimes \mathcal{H}_B$ for finite-dimensional Hilbert spaces $\mathcal{H}_A$ and $\mathcal{H}_B$, a quantum measurement $\{M^{x}_{a}\}_{a \in \cA}$ for each $x \in \cX$ over $\mathcal{H}_A$, and a quantum measurement $\{M^{y}_{b}\}_{b \in \cB}$ each $y$ over $\mathcal{H}_B$. On input $(x,y)$ drawn from $\pi$, Alice performs her measurement corresponding to $x$ on her portion of $\ket{\psi}$, yielding an outcome $a$. Similarly, Bob performs his measurement corresponding to $y$ on his portion of $\ket{\psi}$, yielding an outcome $b$. The results $a$ and $b$ are sent back to the referee, who evaluates if two answers and questions satisfy the predicate $\mathsf{V}$. In summary, the winning probability for the quantum strategy above is given by
\[
    \sum_{x,y} \pi(x,y) \sum_{a,b} \bra{\psi} M_a^x \otimes M_b^y \ket{\psi} \sV(a,b|x,y) \enspace.
\]

The quantum value of a game $G$, denoted by $\omega_q(G)$, is the supremum of the winning probabilities over all the quantum strategies of Alice and Bob.

We consider the parallel repetition of a non-local game $G$, which is a game where the parties try to win $N$ copies of $G$ simultaneously, and the parallel repetition game is denoted by $G^{\otimes N}$. In the game $G^{\otimes N}$, we consider questions $x = (x_1,\dots,x_N) \in \cX_N$, $y=(y_1,\dots,y_N) \in \cY_N$, where each pair $(x_i, y_i) \in \cX \times \cY$ is chosen independently according to the original distribution $\pi$. For outputs $a = (a_1,\dots,a_N) \in \cA_N$ and $b = (b_1,\dots,b_N) \in \cB_N$, Alice and Bob win the parallel repetition game if they win simultaneously on all $n$ coordinates, that is, if for every $i \in [N]$, we have $\sV(x_i,y_i|a_i,b_i) = 1$.

For classical values of any two-player non-local game, an exponential decay of values with parallel repetition is known.
\begin{lemma}[Parallel repetition theorem \cite{raz1998parallel,hol2009parallel}]\label{lem:parallel_repetition}
    For any two-prover game $G = (\pi, \cX \times \cY, \cA \times \cB, \mathsf{V})$ such that $\omega_c(G) = 1-\eps$ for $0<\eps\leq \frac{1}{2}$, 
    \[
        \omega_c(G^{\otimes N}) \leq (1-\eps^c)^{\Omega(\frac{N}{s})} \enspace,
    \]
    where $s=\log |\cA \times \cB| +1$ and $c$ is a universal constant.
\end{lemma}

\subsection{Interactive proof systems}

An $\AM$ protocol is a single-prover interactive proof ($\IP$) protocol where the number of rounds is one and the first message sent by the verifier is sampled uniformly at random.

\begin{definition}[$\AM$]
    We say a language $L$ is in $\AM$ if and only if for an input $z \in \{0,1\}^*$, there exists a single-prover one-round interactive proof system satisfying the following properties:
    \begin{enumerate}
        \item Verifier's messages are sampled uniformly at random from $\{0,1\}^{p(|z|)}$ for some polynomial $p$.
        \item If $z \in L$, the verifier accepts with probability at least $\frac{2}{3}$. 
        \item If $z \notin L$, the verifier accepts with probability at most $\frac{1}{3}$.
    \end{enumerate}
\end{definition}

\begin{fact}[\cite{goldwasser1986private,babai1988arthur}]
    Any language recognized by single-prover constant-round interactive proof systems is in $\AM$.
\end{fact}

\begin{definition}[$\NEXP$]
    We say a language $L$ is in $\NEXP$ if and only if there exist polynomials $p,q$ and a deterministic classical algorithm $\mathcal{A}$ satisfying the following properties:
    \begin{enumerate}
        \item For all $x,y$ the algorithm $\mathcal{A}$ runs in time $2^{p(x)}$ on input $(x,y)$.
        \item If $x \in L$, there exists a string $y$ of length $2^{q(|x|)}$ such that the algorithm $\mathcal{A}$ given $x$ and $y$ accepts.
        \item If $x \notin L$, then for all strings $y$ of length $2^{q(|x|)}$, the algorithm $\mathcal{A}$ given $x$ and $y$ rejects.
    \end{enumerate}
\end{definition}

We recall the definition of the complexity class $\RE$, which stands for the set of recursively enumerable languages. We say a language $L \subseteq \{0,1\}^*$ is in $\RE$ if
and only if there exists a Turing machine $M$ such that if $x \in L$, then $M(x)$ halts and outputs $1$, and if $x \notin L$, then either $M(x)$ outputs $0$ or it does not halt. The Halting Problem is the language that contains descriptions of Turing machines that halt on the empty input. It is known that the Halting Problem is complete for $\RE$, meaning that every language $L \in \RE$ can be reduced in polynomial-time to the Halting Problem (see Lemma 12.8 in \cite{ji2020mip*}).

Let us define the class $\MIP$ using classical values of non-local games. In the definition (and throughout this paper), we consider a two-prover and one-round $\MIP$ protocol while more generally, we can consider $\poly(n)$-prover, $\poly(n)$-round $\MIP$ protocols.

\begin{definition}[$\MIP$]\label{def:mip}
    We say a language $L$ is in $\MIP_{c,s}[q,a]$ if and only if, for $z \in \{0,1\}^*$, there exists a family of non-local games $G_z = (\pi, \cX \times \cY, \cA \times \cB, \mathsf{V})$ satisfying the following properties: 
    \begin{enumerate}
        \item The question size is $q$ and the answer size is $a$.
        \item There is an algorithm running in $\poly(|z|)$-time, given $z$, to sample questions $(x,y) \in \cX \times \cY$ from $\pi$.
        \item There is an algorithm running in $\poly(|z|)$-time, given $z$ and $(x,y,a,b) \in \cX \times \cY \times \cA \times \cB$, to output $\sV(a,b|x,y) \in \{0,1\}$.
        \item If $z \in L$, $\omega_c(G_z) = c$
        \item If $z \notin L$, $\omega_c(G_z) \leq s$.
    \end{enumerate}
\end{definition}

When we do not specify the parameters of $\MIP$, the completeness and soundness parameters are $c=1$ and $s=\frac{1}{2}$, and $q$ and $a$ are some $\poly(n)$.

\begin{fact}[\cite{babai1991mip}]
    $\MIP = \NEXP$.
\end{fact}

In $\MIP^*$ protocols, the provers can share some entangled quantum states before the protocol starts. Below, we define $\MIP^*$ protocols using quantum values of non-local games.

\begin{definition}[$\MIP^*$, Definition 5.29 in \cite{ji2020mip*}]\label{def:mip*}
    We say a language $L$ is in $\MIP^*$ if and only if, for an input $z \in \{0,1\}^*$, there exists a family of non-local games $G_z = (\pi, \cX \times \cY, \cA \times \cB, \mathsf{V})$ satisfying the following properties: 
    \begin{enumerate}
        \item There is an algorithm that runs in $\poly(|z|)$-time, given $z$, to sample questions $(x,y) \in \cX \times \cY$ from $\pi$.
        \item There is an algorithm that runs in $\poly(|z|)$-time, given $z$ and $(x,y,a,b) \in \cX \times \cY \times \cA \times \cB$, to output $\sV(a,b|x,y) \in \{0,1\}$.
        \item If $z \in L$, $\omega_q(G_z) = 1$.
        \item If $z \notin L$, $\omega_q(G_z) \leq \frac{1}{2}$.
    \end{enumerate}
\end{definition}

$\MIP^*$ protocols are known to be extremely powerful, recognizing the Halting Problem, which is $\RE$-complete. 
\begin{theorem}[\cite{ji2020mip*}]
    $\MIP^*=\RE$.
\end{theorem}

\subsection{Probabilistically checkable proofs}

Probabilistically checkable proofs (PCPs) are one of the most important results in the study of complexity theory, with many applications ranging from hardness of approximation to construction of secure cryptographic protocols. In this paper, we view PCPs as constructions that allow us to \emph{amplify} the hardness of arbitrary 3-SAT instances. We will start by defining the label cover problem and proceed with recalling some well-known and recent results in the PCP literature.

\begin{definition}[Label cover]\label{def:label-cover}
    An instance of label cover $\Psi = (G=(L\cup R,E), \Sigma_L, \Sigma_R, \Phi = \{\Phi_{e}\}_{e\in E})$ consists of a bipartite graph $G$, alphabets $\Sigma_L, \Sigma_R$ and constraints $\Phi_e\subseteq \Sigma_L\times \Sigma_R$, one for each edge.
    Each one of the constraints is a projection constraint, meaning that for every $e\in E$ there is a map $\phi_e\colon \Sigma_L\to\Sigma_R$ such that
    \begin{equation*}
        \Phi_e = \{(\sigma,\phi_e(\sigma))~|~\sigma\in \Sigma_L\} \enspace.
    \end{equation*}
\end{definition} 

Given a label cover instance $\Psi$, the goal is to find assignments $A_L\colon L\to\Sigma_L$ and $A_R\colon R\to\Sigma_R$ that satisfy as many of the constraints as possible, namely that maximize the quantity $\val(\Psi) = \max_{A_L,A_R}\val_{\Psi}(A_L,A_R)$ where
\begin{align*}
    \val_{\Psi}(A_L,A_R) = \frac{1}{|E|}\left|\{e=(u,v)\in E~|~(A_L(u),A_R(v))\in \Phi_e\}\right|\enspace.
\end{align*}

We can view the PCP theorem as showing that the problem $\textsc{GapLabelCover}[c,s]$ is $\NP$-complete, where $c-s=\Omega(1)$ is a constant. However, there is still interest in optimizing the parameters of interest:
\begin{itemize}
    \item \emph{Completeness:} The parameter $c$, and ideally we want to be as high as possible. We will always work with the case where $c=1$, which is denoted by perfect completeness.
    \item \emph{Soundness:} The parameter $s$, which we ideally want to have as small as possible.
    \item \emph{Instance size:} Different constructions of PCPs start with an arbitrary instance of 3-SAT of size $n$, and show a polynomial-time reduction to an instance of label cover. The overhead in the size of the label cover instance as a function of $n$ is of major interest and is denoted by instance size.
\end{itemize}

We will proceed with stating some of the known PCP constructions. In this view, the original PCP theorem of \cite{arora1998probabilistic,arora1998proof} can be stated as follows.

\begin{lemma}[\cite{arora1998probabilistic,arora1998proof}]
    $\textsc{GapLabelCover}[1,1-\eps]$ for a tiny (but constant) $\eps$ is $\NP$-complete, and the instance size is $\poly(n)$.
\end{lemma}

Also, we state Dinur's PCP theorem \cite{dinur2007pcp}, which utilizes a combinatorial proof using gap amplification as follows.

\begin{lemma}[\cite{dinur2007pcp}]
    $\textsc{GapLabelCover}[1,1-\Omega(1)]$ is $\NP$-complete, and the instance size is $\poly(n)$.
\end{lemma}

Note that here the improvement is in the soundness parameter, where we have control over how small the soundness error will be, by choosing the parameters of gap amplification. However, it is important to note that this technique also has some limitations. It was shown by Bogdanov \cite{bogdanov2005gap} that one cannot hope to reduce the soundness error to a constant below $1/2$. 

An original motivation for studying parallel repetition results \cref{lem:parallel_repetition} was to achieve an arbitrarily small soundness error (and in particular, a soundness error below $1/2$). The cost, however, appears in the instant size, as it is required that we increase the instant size and the alphabet size polynomially, where the polynomial depends on the number of copies of the original instance.

Finding constructions of PCPs where the overhead in instance size is as small as possible is one of the main directions of studying PCP theorems. Recently, Bafna et al.\! \cite{bafna2025quasi} proved a reduction to construct a label cover instance from a given 3-SAT with significantly better parameters, improving previous work.

\begin{lemma}[Theorem 1.2 in \cite{bafna2025quasi}]\label{lem:PCP2025}
    Let $\phi$ be a 3-SAT instance of size $n$, and let $s > 0$ be a constant. Then, there exists a constant $c=c(s)$ and a label cover instance $\Psi$ of size at most $n(\log n)^c$ and alphabet size at most $\poly(1/k)$ such that
    \begin{itemize}
        \item if $\phi$ is satisfiable, then $\val(\Psi)=1$,
        \item and if $\phi$ is unsatisfiable, then $\val(\Psi)\leq s$.
    \end{itemize}
\end{lemma}
An important property of their construction, which is proved by using properties of high-dimensional expanders, is that not only is the blowup in size quasi-linear, but also one can set the soundness error to be any arbitrary small value (while paying a moderate price in the instant size and alphabet size). 

Another result that we will rely on in this work is based on the work of Dinur, Harsha, and Kindler \cite{dinur2015lowerrorPCP}, which we state as the next lemma.

\begin{lemma}[Theorem 1.4 in \cite{dinur2015lowerrorPCP}]\label{lem:low-error-PCP}
Every language in $\NP$ has a PCP verifier that tosses $O(\log n)$ random bits, makes $(\log \log n)^{O(1)}$ queries into a proof over an alphabet $\Sigma$ of size $|\Sigma|=n^{1/(\log\log n)^{O(1)}}$, has perfect completeness and soundness error $1/\poly(n)$.
\end{lemma}

This result is different in a sense from the ones stated before, since the focus is not on keeping the query complexity constant. Instead, what they focus on and prove is inverse polynomial soundness error while keeping the query complexity as small as possible. This is not known to be achievable in the constant query regime\footnote{This problem is known as the Sliding Scale Conjecture. We will discuss it in \cref{sec:conclusion}.}.

We now turn our attention to the view of the class of problems admitting a PCP proof and a 2-query verifier. We can state the class of problems with such properties as follows.
\begin{definition}[PCP class]\label{def:PCP}
    Let $\PCP_{1,s}[r,2]_\Sigma$ be the class of languages that admit a 2-query verifier that queries a proof over alphabet $\Sigma$ and has with perfect completeness, soundness $s$, and randomness complexity $r$.
\end{definition}

An immediate corollary of \cref{lem:PCP2025} stated in terms of \cref{def:PCP} is the following.
\begin{corollary}\label{cor:good-PCP}
    For every positive $s > 0$, we have $\NP \subseteq \PCP_{1,s}[\log n + O_s(\log\log n),2]_\Sigma$, where $|\Sigma| =O_s(1)$. Here, $O_s(\cdot)$ hides a constant that only depends on $s$.
\end{corollary}

Throughout this work, we will use the known fact that there are equivalences between 1-round 2-prover multi-prover interactive proofs and PCPs. One equivalence is as follows.

\begin{fact}\label{fact:PCPgame direct}
    $\PCP_{1,s}[r,2]_\Sigma = \MIP_{1,s}[r,\log (|\Sigma|)]$.
\end{fact}
For clarity, let us describe the property above more formally. Let us consider a verification protocol to sample an edge $e=(u,v)$, send each vertex to each prover, and check that two assignments satisfy the constraint.

The other is from the consistency test.

\begin{fact}\label{fact:PCPgame}
    $\PCP_{1,s}[r,2]_\Sigma = \MIP_{1,(1+s)/2}[2r,2\log (|\Sigma|)]$.
\end{fact}

This equivalence can be showed in the following way: Considering the label cover instance, the verifier samples an edge $e=(u,v)$ and asks the first prover for the assignments to its endpoints $u$ and $v$. Then, it will choose one endpoint uniformly at random, and asks the second prover for the assignment to that endpoint. Here, the second question serves as a consistency test to make sure that the first and second provers are answering with respect to one fixed assignment that they have in their minds. Indeed, this equivalence heavily relies on the fact that there is no communication between the provers, and it is unclear what happens if limited communication is allowed.

By scaling up the PCP theorem, the following succinct property of $\MIP$ protocols for $\NEXP$ is also known.

\begin{fact}\label{fact:MIP=NEXP succint}
    $\MIP_{1,\frac{1}{2}}[\poly(n), O(1)] = \NEXP$.
\end{fact}

%% file: proof.tex
\section{Proofs for main results}
In the next part, we formally define the notions of multi-prover interactive proofs with leakage, and then proceed to show our results on the power of these complexity classes.

\subsection{Definitions of $\MIP$ and $\MIP^*$ protocols with leakage}\label{sec:def of MIP with leakage}

Let us define a $\MIP[q,a,\ell]$ protocol as a two-prover one-round $\MIP$ protocol with question size $q$ and answer size $a$, where the provers conduct $\ell$-bit communication between themselves based on their questions received from the verifier. The provers then send some answers back to the verifier based on their communication and questions. We also use $\MIP[q,a,\ell]$ as the set of languages recognized by the $\MIP[q,a,\ell]$ protocols. A formal definition follows.

\begin{definition}[$\MIP$ with leakage]\label{def:mip-leakage}
    We say a language $L$ is in $\MIP_{c,s}[q,a, \ell]$ if and only if, for $z \in \{0,1\}^*$, there exists a family of non-local games $G_z = (\pi, \cX \times \cY, \cA \times \cB, \mathsf{V})$ satisfying the following properties: 
    \begin{enumerate}
        \item The question size is $q$ and the answer size is $a$.
        \item There is an algorithm running in $\poly(|z|)$-time, given $z$, to sample questions $(x,y) \in \cX \times \cY$ from $\pi$.
        \item There is an algorithm running in $\poly(|z|)$-time, given $z$ and $(x,y,a,b) \in \cX \times \cY \times \cA \times \cB$, to output $\sV(a,b|x,y) \in \{0,1\}$.
        \item If $z \in L$, $\omega_c(G_z) = c$.
        \item If $z \notin L$, and even if the provers are allowed to communicate at most $\ell$ bits, we must have $\omega_c(G_z) \leq s$.
    \end{enumerate}
\end{definition}

When we do not specify completeness $c$ and soundness $s$, it should be assumed that $c=1$ and $s={1}/{2}$. 

Similarly, let us define a $\MIP^*[q,a,\ell]$ protocol and $\MIP^*[q,a,\ell]$ as the set of languages recognized by $\MIP^*[q,a,\ell]$ protocols. In a $\MIP^*[q,a,\ell]$ protocol, the provers can exploit the shared entanglement for their communication.

\begin{definition}[$\MIP^*$ with leakage]\label{def:mip*-leakage}
    We say a language $L$ is in $\MIP^*[q,a,\ell]$ if and only if, for an input $z \in \{0,1\}^*$, there exists a family of non-local games $G_z = (\pi, \cX \times \cY, \cA \times \cB, \mathsf{V})$ satisfying the following properties: 
    \begin{enumerate}
        \item The question size is $q$ and the answer size is $a$.
        \item There is an algorithm running in $\poly(|z|)$-time, given $z$, to sample questions $(x,y) \in \cX \times \cY$ from $\pi$.
        \item There is an algorithm running in $\poly(|z|)$-time, given $z$ and $(x,y,a,b) \in \cX \times \cY \times \cA \times \cB$, to output $\sV(a,b|x,y) \in \{0,1\}$.
        \item If $z \in L$, $\omega_q(G_z) = 1$.
        \item If $z \notin L$, and even if the provers are allowed to communicate at most $\ell$ bits, $\omega_q(G_z) \leq \frac{1}{2}$.
    \end{enumerate}
\end{definition}

Note that for $\ell \geq q$, $\MIP[q,a,\ell] \subseteq \AM$ and $\MIP^*[q,a,\ell] \subseteq \AM$ because the provers can share their questions by communication and the game value becomes trivial.

We also define a notion of multi-prover interactive proofs, in which the communication between the provers is one-way. In this scenario, the verifier asks a question from the first prover, then the first prover answers the question to the verifier and is allowed to send a bounded number of bits to the second prover. After that, the verifier will ask the second question from the second prover. This time, instead of using the terminology of non-local games, we give the definition in terms of the task of verifying a decision problem. This will make the corresponding results simpler to present. The formal definition follows.

\begin{definition}\label{def:oneway-leakage}
    A language $L$ is in $\MIP^{\rightarrow}_{c,s}[q,a,\ell]$ if and only if there exists a two-prover one-round interactive proof for $L$, in which the verifier $V$ runs in time $\poly(|x|)$ for every instance $x\in\{0,1\}^*$ and satisfies the following properties:
    \begin{itemize}
        \item \emph{Completeness:} If $x\in L$, then after exchanging one round of messages first with $P_1$ and then with $P_2$, the verifier $V$ accepts with probability at least $c$.
        \item \emph{Soundness:} If $x\notin L$, for every set of malicious provers $P_1$ and $P_2$, in which $P_1$ is allowed to send $\ell$ bits to $P_2$ after receiving the question from $V$, the verifier will accept with probability at most $s$.
    \end{itemize}
    Additionally, the question size is $q$ bits, and the answer size is at most $a$ bits.
\end{definition}

\subsection{The power of $\MIP$ with leakage}\label{sec:MIP}
We say a communication protocol is a (randomized) public coin communication protocol if the parties share some randomness before they receive inputs. A public coin communication protocol can be regarded as a convex mixture of deterministic protocols over the probabilistic distribution of public coins.

The following statement follows from the parallel repetition theorem. A similar statement is known in previous work (see, e.g., Corollary 7.3 in \cite{jain2022direct} and Theorem 8 in \cite{hasegawa2025maximum}).

\begin{lemma}\label{lem:classical parallel repetition theorem with leakage}
    Let $G$ be a non-local game such that $\omega_c(G) = 1-\eps$ for a constant $\eps>0$ and the answer size is constant. Let $\mathcal{P}$ be a public-coin $2$-way classical communication protocol for $G^{\otimes N}$ and suppose that the communication amount of $\mathcal{P}$ is $cN$ for a sufficiently small constant $c>0$. Then, the success probability of $\mathcal{P}$ is at most $2^{-\Omega(N)}$.
\end{lemma}

\begin{proof}
    Let $c>0$ be a sufficiently small constant and $\mathcal{P}$ be a public-coin classical 2-way communication protocol for $G^{\otimes N}$ whose communication amount is $c N$. Let $p$ be the success probability of $\mathcal{P}$. 

    Let us consider a communication protocol $\mathcal{P}'$ in which the parties share uniformly random $cN$ bits in addition to the public coins, but have no communication. In the beginning of the protocol, Alice and Bob check that their random bits are equal to the $cN$ bits of the communication protocol, and if not, they abort. Let $\not \perp_A$ be the event that Alice does not abort during the protocol $\mathcal{P}$, and let $\not \perp_B$ be the event that Bob does not abort during the protocol $\mathcal{P}$.

    Define $W_{\mathcal{P}'}$ to be the event that $\mathcal{P}'$ outputs wins all the games (hence, without abortions). Then, from the definition of $\mathcal{P}'$, we have
    \[
        \mathrm{Pr} [ W_{\mathcal{P}'} | \not \perp_A, \not \perp_B] = p \enspace.
    \]
    Since randomness is sampled uniformly and independently from the public coins of the communication protocols, we also have 
    \[
        \mathrm{Pr} [ \not \perp_A, \not \perp_B ] = 2^{-c N}\enspace,
    \]
    which implies 
    \[
        \mathrm{Pr} [ W_{\mathcal{P}'} ] = 2^{-cN} p\enspace.
    \]
    Considering abortions, the success probability of $\mathcal{P}'$ is at least $2^{-cN} p$.
    From \cref{lem:parallel_repetition}, and since shared randomness does not change the value of the games, the success probability of $\mathcal{P}'$ is also at most $2^{-{c'}N}$ for a constant $c'$. Taking $c>0$ such that $c < c'$ and combining the upper and lower bound of the success probability, $p \leq 2^{-(c'-c)N}$.
\end{proof}

\begin{theorem}\label{thm:mip-leakage-NP}
    For any $\ell \in \mathbb{N}$ such that  $\ell \le \poly(n)$, 
    \[
        \MIP[O(\ell\log n),O(\ell),\ell] \supseteq \NP\enspace.
    \]
\end{theorem}

\begin{proof}
    Let $G$ be a non-local game corresponding to the $\MIP$ protocol for $\NP$ in \cref{fact:PCPgame}. Then, let $k \in \mathbb{N}$ be a sufficiently large constant and consider a $\MIP$ protocol corresponding to $G^{\otimes k\ell}$ ($(k\ell)$-wise parallel repetition of $G$). Since perfect completeness holds in \cref{fact:PCPgame}, it also holds in the parallel repetition game. Applying \cref{lem:classical parallel repetition theorem with leakage} will prove the soundness also holds. The question and answer size of $G^{\otimes k\ell}$, given that $k$ is chosen to be some large constant, are $O(\ell \log n)$ and $O(\ell)$ respectively.
\end{proof}

\begin{theorem}\label{thm:mip leakage}
    For any $\ell \in \mathbb{N}$ such that  $\ell \le \poly(n)$,  
    \[
        \MIP[O(\ell) \cdot \poly(n),O(\ell),\ell] \supseteq \NEXP \enspace.
    \]
\end{theorem}

\begin{proof}
    Let $G$ be a non-local game corresponding to the $\MIP$ protocol for $\NEXP$ in \cref{fact:MIP=NEXP succint}. Then, let $k \in \mathbb{N}$ be a sufficiently large constant and consider a $\MIP$ protocol corresponding to $G^{\otimes k\ell}$. Since perfect completeness holds in \cref{fact:PCPgame}, it also holds by parallel repetition. Applying \cref{lem:classical parallel repetition theorem with leakage} will show that the soundness also holds. The question and answer size of $G^{\otimes k\ell}$, given that $k$ is chosen to be some large constant, are $O(\ell) \cdot \poly(n)$ and $O(\ell)$ respectively.
\end{proof}

\subsection{The power of $\MIP^*$ with leakage}\label{sec:MIP*}

Next, we consider $\MIP^*$ protocols with leakage. Jain and Kundu \cite{jain2022direct} showed the following direct product theorem for the entanglement-assisted quantum communication complexity. 

\begin{lemma}[Part of Theorem 1.1 in \cite{jain2022direct}]\label{lem:quantum dpt}
For any predicate $\mathsf{V}$ on $(\cA\times\cB)\times(\cX\times\cY)$ and any product probability distribution $\pi$ on the question set $\cX\times\cY$, let $\mathcal{P}$ be an interactive entanglement-assisted 2-party quantum communication protocol for $G^{\otimes N}$. If $\mathcal{P}$ has total communication $cN$ for $c < 1$, the success probability of $\mathcal{P}$ is at most
\[ 
    \left(1-\frac{\nu}{2} + 2\sqrt{c}\right)^{\Omega\left(\frac{\nu^2N}{\log(|\cA|\cdot|\cB|)}\right)}
\]
where $\nu=1-\omega_q (G)$.
\end{lemma}

The above statement can also be regarded as a parallel repetition theorem for the quantum value with leakage.

Natarajan and Zhang \cite{natarajan2023quantum} showed that there exists a quantum free game protocol for $\RE$. Denote a two-prover one-round $\MIP^*$ protocol as an $\AM^*(2)$ protocol if the question distribution is uniform. Note that uniform distribution is a special case of product distributions.

\begin{lemma}[\cite{natarajan2023quantum}, see also Theorem 2.2 in \cite{mastel2024two}]\label{lem:qfg}
    There exists an $\AM^*(2)$ protocol for the Halting Problem with constant length questions and $\poly\log (n)$ length answers, perfect completeness and soundness $\frac{1}{2}$.
\end{lemma}

\begin{theorem}\label{thm:mip* with leakage}
    There exist constants $c$ and $c'$ such that for any $\ell  \le \poly(n)$,
    \[
        \MIP^*[O(\max\{\ell, \log^c n\}), O(\max\{\ell, \log^c n\} \cdot \log^{c'} n ),\ell] \supseteq \RE \enspace.
    \]
\end{theorem}

\begin{proof}
    Let $G$ be a non-local game corresponding to the $\MIP^*$ protocol for $\RE$ in \cref{lem:qfg}. Then, let $c$ and $k$ be sufficiently large constants and consider a $\MIP^*$ protocol corresponding to $G^{\otimes k \cdot \max\{\ell, \log^c n\}}$. Since perfect completeness holds in \cref{lem:qfg}, it also holds in the parallel repetition game. Applying \cref{lem:quantum dpt} will prove that soundness\footnote{To be precise, we here consider soundness against $\ell$-bit communication which can be simulated using $l$-qubit communication.} also holds. The question size of $G^{\otimes k \cdot \max\{\ell, \log^c n\}}$ is $O(\max\{\ell, \log^c n\})$ and the answer size is $O(\max\{\ell, \log^c n\} \cdot \log^{c'} n )$ for some $c'>0$ from \cref{lem:qfg}.
\end{proof}

\subsection{Low-soundness PCPs and one-way leakage}\label{sec:PCP2025}

We have seen from the proof of \cref{lem:classical parallel repetition theorem with leakage} that we can allow $\ell$ bits of leakage as long as we have a soundness parameter $s \leq 2^{-\ell}$. Therefore, from \cref{fact:PCPgame direct}, we have the following corollary.

\begin{corollary}\label{cor:immediate}
$\MIP_{1,s \cdot 2^{\ell}}[\log n + O(\log\log n), O(1),\ell] \supseteq \NP\enspace.$
\end{corollary}
However, this method only works for 2-query PCPs. In this section, we present a more general method for converting a PCP into a $\MIP$ protocol that is resilient to leakage. First, we prove the following theorem for arbitrary CSPs.

\begin{theorem}\label{thm:CSP-to-MIP}
    Let $\ell$ be a positive constant, and $\Psi$ be a CSP instance of size $\poly(n)$ on $n$ variables, such that every constraint is of arity $k$ over an alphabet $\Sigma$. Suppose $\Psi$ is promised to be either fully satisfiable or any assignment to the variables will satisfy at most $\frac{1}{2^{\ell+1}}$-fraction of its constraints. Then,
    \begin{align*}
        \Psi \in \MIP^{\rightarrow}_{1,1-\frac{1}{2k}} \left[ O(\log n), k \log |\Sigma|,\ell \right] \enspace.
    \end{align*}
\end{theorem}

\begin{proof}
    The verifier for the satisfiability of the CSP instance $\Psi$ works as follows.
    \begin{tcolorbox}[title=Algorithm~{1}: $\MIP$ protocol for verifying the satisfiability of $\Psi$]   
    \paragraph{The verifier $V$ having access to $\Psi$.}
    \begin{enumerate}
        \item Sample a constraint $e=(u_1,...,u_{k})$ of $\Psi$ uniformly at random and send it $P_1$.
        \item After receiving the assignments $a(u_1),...,a(u_k)$ sample a random element $i\in [k]$.
        \item Ask $P_2$ for the assignment to $u_i$.
        \item Accept if and only if the answers of the two provers are consistent and if the assignment satisfies the constraint.
    \end{enumerate}
\end{tcolorbox}

In the protocol above, the question size is $O(\log n)$ because the size of CSP is $\poly(n)$, and thus the verifier can specify a constraint and a vertex with $O(\log n)$ bits.

It is clear that if $\Psi$ is satisfiable, then the verifier accepts after interacting with honest provers.
Now suppose that $\Psi$ is unsatisfiable, and in particular, by the assumption of the theorem, we know that at any assignment will satisfy at most $2^{-(\ell+1)}$-fraction of the constraints. Therefore, for \emph{every set of $2^\ell$ assignments}, $1/2$ of the constraints are not satisfied by all of them.

We can assume the provers are deterministic. Therefore, $P_2$ upon receiving the leakage $\ell$\footnote{Here, we are slightly abusing the notation and using $\ell$ for the leakage content as well as its length.} fixes its strategy, which we denote by $P_2^\ell$. Note that for each value of $\ell$, the strategy $P_2^\ell$ corresponds to an assignment to the variables. 
Hence, regardless of the set of assignments the provers agree on before the protocol starts, a $\frac{1}{2}$-fraction of constraints will remain unsatisfied. Therefore, in the first step, with probability $1/2$, the verifier samples such a constraint and forces $P_1$ to lie. After $P_2$ receives the leakage, there is at least $1/k$ chance that the verifier samples a variable that $P_2$ will answer inconsistently with respect to $P_1$'s answer. Hence, the soundness is at most $1-1/2k$, which concludes the proof.
\end{proof}

\subsection{Concrete examples}
Having \cref{thm:CSP-to-MIP}, we can apply the transformation to two different PCP constructions to get concrete bounds.

\subsubsection{Verifying $\NP$ with $O(1)$ bits of leakage}
Our first example is to use the recent construction of \cite{bafna2025quasi}. Before stating the result, we show how to construct an $\MIP$ protocol for the following promise label cover problem. 

\begin{lemma}\label{lem:leakage}
    Suppose $\Psi$ is a given label cover of size $S$ and alphabet $\Sigma$ such that either $\val(\Psi)=1$ or $\val(\Psi) \leq \frac{1}{2^{\ell+1}}$ for $\ell \in \mathbb{N}$. Then, there exists a $\MIP$ protocol for verifying the satisfiability of $\Psi$ in which the verifier $V$ has the following properties:
    \begin{itemize}
        \item Perfect completeness: $V$ accepts with probability 1 if $\val(\Psi)=1$.
        \item Soundness against $\ell$ bit of leakage: $V$ accepts with probability at most $\frac{3}{4}$ even if $\ell$ bit of one-way leakage among the cheating provers is allowed.
    \end{itemize}  
\end{lemma}

\begin{proof}
The proof follows immediately from \cref{thm:CSP-to-MIP}, since label cover is a CSP of arity 2.
\end{proof}

Having this claim, we are now ready to prove the following theorem.
\begin{theorem}\label{thm:lbitPCP}
    For any constant $\ell \in \mathbb{N}$,
    \[
        \MIP^{\rightarrow}_{1,\frac{3}{4}}[2\log n + O_\ell(\log\log n), O_\ell(1),\ell] \supseteq \NP \enspace.
    \]
\end{theorem}

\begin{proof}
     Let $\phi$ be any 3-SAT instance. From \cref{lem:PCP2025}, for any constant $\ell \in \mathbb{N}$, there exists a label cover instance $\Psi$ of size at most $n (\log n)^c$ for some constant $c(\ell)$ and alphabet size that depends only on $\ell$ such that $\val(\Psi)=1$ if $\phi$ is satisfiable and $\val(\Psi)<\frac{1}{2^{\ell+1}}$ if $\phi$ is not satisfiable. The theorem follows by applying \cref{lem:leakage} to $\Psi$.
\end{proof}

\subsubsection{Verifying $\NP$ with $O(\log n)$ bits of leakage}

In this part, we turn our attention to a particular setting of interest in $\MIP$ protocols for $\NP$, where we allow the provers to have $O(\log n)$ bits of one-way communication, and our goal is to optimize the question and answer size. Recall that from \cref{thm:mip-leakage-NP}, we can set $\ell = O(\log n)$ and get $\NP \subseteq \MIP[O(\log^2 n), O(\log n), O(\log n)]$. Here, we prove a similar result, albeit with different parameters, which we will discuss shortly. The main theorem of this section is the following.

\begin{theorem}\label{thm:logn-leakage}
    Let $\ell$ be a positive constant. Then,
    \begin{align*}
        \MIP^{\rightarrow}_{1,1-\frac{1}{\poly(\log\log n)}} \left[O(\log n), {\log n}\cdot {\poly(\log\log n)},\ell \log(n)\right] \supseteq \NP \enspace.
    \end{align*}
\end{theorem}

Before moving on to proving the theorem, we would like to emphasize that this statement differs from what we get from \cref{thm:mip-leakage-NP} in two aspects: The question size remains $O(\log n)$ (compared to $O(\log^2 n)$) and the answer size is $\log n \cdot \poly(\log\log n)$ (compared to $O(\log n)$). The main downside is, however, that the soundness error will not remain constant.

We will start by proving the following claim.

\begin{claim}\label{clm:CSP-to-MIP}
    Let $\ell$ be a positive constant, and $\Psi$ be CSP instance of size $\poly(n)$ on $n$ variables, such that every constraint is of arity $\poly(\log\log n)$, and the alphabet size is $n^{1/\poly(\log\log n)}$. Suppose $\Psi$ is promised to be either fully satisfiable or at most $\frac{1}{2n^\ell}$-fraction of its constraints is satisfied. Then,
    \begin{align*}
        \Psi \in \MIP^{\rightarrow}_{1,1-\frac{1}{\poly(\log\log n)}} \left[ O(\log n), {\log n}\cdot {\poly(\log\log n)},\ell \log(n) \right] \enspace.
    \end{align*}
\end{claim}
\begin{proof}
The proof follows immediately from \cref{thm:CSP-to-MIP}.
\end{proof}

\begin{proof}[Proof of \cref{thm:logn-leakage}]
    Let $\phi$ be an arbitrary 3-SAT instance on $n$ variables. By \cref{lem:low-error-PCP}, there exists a PCP proof of size $\poly(n)$ over an alphabet of size $n^{1/(\log\log n)^{O(1)}}$ such that a verifier with $(\log\log n)^{O(1)}$ queries can use to reject if $\phi$ is unsatisfiable with probability $1-\frac{1}{\poly(n)}$.

    Consider the CSP corresponding to the queries and the verification algorithm of the verifier. This CSP satisfies all the constraints of \cref{clm:CSP-to-MIP} (if the rejection probability of the verifier is not high enough, we can simply apply sequential repetition to make sure it is more than $1-\frac{1}{2n^\ell}$). This concludes the proof. 
\end{proof}

%% file: summary.tex
\section{Concluding remarks}\label{sec:conclusion}
In this paper, we investigated the problem of finding the complexity of multi-prover interactive proof systems where the provers are allowed to communicate bounded amount of messages to each other. We proposed and developed two techniques, one based on the parallel repetition results and the other based on gap amplification and low-soundness PCPs, to address the problem. In summary, using these techniques, we showed that we can still maintain the power of multi-prover interactive proof systems by increasing the question and answer sizes moderately. These results show that in particular, for real-world applications, we may not need to enforce spatial separation between provers to cross-examine them. Instead, we can still keep the verification power as long as we have a predetermined bound on the amount of communication they are allowed to have. We also showed an interesting connection between this scenario and the Sliding Scale Conjecture in the PCP literature, hoping that this might be an novel approach for making progress on that front. In addition, there are three ways one can interpret our results:

\paragraph{Adaptiveness in $\MIP$s:} It is a well-known fact that the power of multi-prover interactive proofs, and in particular, one-round $\MIP$s lies in the fact that the provers' answers are non-adaptive. Essentially, this follows from the fact that each prover is blind to the question the other prover receives. At the same time, if we allow complete adaptiveness between the messages, we are back to the power of single-prover interactive proofs, which is incomparable to the former ($\AM$ vs. $\NEXP$). However, one could ask, \emph{how much adaptiveness} is required for the cheating provers to be able to convince the verifier to accept an incorrect statement? Could it be the case that a single bit of adaptiveness is sufficient? Our results address this problem negatively. In a way, we show that if the amount of adaptiveness is bounded, then the power of multi-prover interactive proofs remains intact. However, there is more to be done, as it is still unclear where this transition in complexity appears.

\paragraph{Another application of PCPs:} PCPs undoubtedly have many interesting applications in complexity theory and cryptography, and our result can also be interpreted as a new one. We showed that one of the applications of reducing soundness error in PCPs, whether through parallel repetition or through novel constructions, can be viewed as mildly relaxing the no-communication assumption in the $\MIP$ view. Whether this can be useful for showing no-go results in PCP constructions or whether this idea can be used in real-world cryptographic constructions is a possibility that we leave for future research. In particular, it would be very appealing to know whether our approach can bring us any closer to finding a resolution for the Sliding Scale Conjecture in the PCP literature.

\paragraph{Proof of quantumness:} Recently, several protocols for proving quantumness have been proposed based on compiled non-local games \cite{kalai2023quantum,natarajan2023bounding,kulpe2025bound}. The core mechanism involves simulating a two-prover game using a single (supposedly quantum) prover by homomorphically encrypting the first question \cite{brakerski2018quantum,mahadev2020classical}. Ideally, an honest quantum prover can respond within the encrypted domain, whereas a dishonest classical prover cannot extract information from the first question to gain an advantage on the second. However, the security of these protocols relies on cryptographic assumptions. Our work provides an alternative approach. In the compiled setting, leakage between provers effectively requires the adversary to store information in memory for use upon receiving the subsequent question. This allows us to frame the compiled non-local game as sound against \emph{bounded-space} adversaries. While our current bounds may not yet yield a practical real-world protocol, given the low cost of classical memory, they demonstrate a theoretical path toward replacing cryptographic assumptions with complexity-theoretic ones, specifically space-complexity bounds.

\subsection{Open problems}\label{sec:open-problems}

Our work raises several open problems, which we leave for future work. We conclude this paper by naming a few of them here. 
\begin{itemize}
    \item An immediate open problem is to find the optimal tradeoff for allowing constant bits (or qubits) of leakage.
    \item Our technique based on low-soundness PCPs requires the additional assumption that the communication is one-way. We believe that this requirement is not necessary and it is a limitation of the technique. It would be interesting to generalize these tools so that they do not rely on this assumption anymore.
    \item Another open problem is to find a succinct $\MIP$ protocol for $\NP$ that has constant soundness against $O(\log n)$ bits of leakage. As mentioned in \cref{sec:overview}, we can achieve this if there is a positive answer to the Sliding Scale Conjecture. However, we do not know whether the other direction holds. Is it possible to show such a result without relying on \cref{conj:sliding-scale}?
\end{itemize}